\documentclass[11pt]{article}
\usepackage{fullpage}

 \usepackage[numbers]{natbib}
\usepackage{amsthm}
\usepackage{nameref}
\usepackage{varioref}
\usepackage{array}
\newcolumntype{P}[1]{>{\centering\arraybackslash}p{#1}}
\usepackage{gensymb}

\usepackage{balance}
\usepackage{tikz}
\usepackage{subcaption}
\usepackage[ruled,vlined,linesnumbered]{algorithm2e}

\usepackage[utf8]{inputenc}  
\usepackage{hyperref,xcolor}
\hypersetup{
colorlinks,
linkcolor={blue},
citecolor={blue},
urlcolor={blue!80!black}
}

\newtheorem{theorem}{Theorem}
\newtheorem{remark}{Remark}
\newtheorem{definition}{Definition}

\newtheorem{lemma}[theorem]{Lemma}

\theoremstyle{definition}

\AtBeginDocument{%
  }

\begin{document}


\title{Optimal Fault-Tolerant Dispersion on Oriented Grids}

\author{
  Rik Banerjee\footnote{Indian Statistical Institute, Kolkata, India, {\em rikarjya@gmail.com}. The work was done while Rik Banerjee was a master's student at ISI, Kolkata.} 
  \and Manish Kumar\footnote{Indian Institute of Technology, Madras, India, {\em manishsky27@gmail.com}. Part of the work was done while Manish Kumar was a PhD student at ISI, Kolkata.}
  \and Anisur Rahaman Molla\footnote{Indian Statistical Institute, Kolkata, India, {\em anisurpm@gmail.com}}
}

\maketitle              

\begin{abstract}
Dispersion of mobile robots over the nodes of an anonymous graph is an important problem and turns out to be a crucial subroutine for designing efficient algorithms for many fundamental graph problems via mobile robots. In this problem, starting from an arbitrary initial distribution of $n$ robots across the $n$ nodes, the goal is to achieve a final configuration where each node holds at most one robot. This paper investigates the dispersion problem on an oriented grid, considering the possibility of robot failures (crashes) at any time during the algorithm's execution. We present a crash-tolerant dispersion algorithm that solves the dispersion problem on an anonymous oriented grid in $O(\sqrt{n})$ time and using $O(\log n)$ bits of memory per robot. The algorithm is optimal in terms of both time and memory per robot. We further extend this algorithm to deal with weak Byzantine robots. The weak Byzantine fault dispersion algorithm takes optimal $O(\sqrt{n})$ rounds but requires $O(n\log n)$ bits of memory per robot. 
\end{abstract}

\section{Introduction}\label{sec: introduction}

The distribution of autonomous mobile robots for achieving coverage across an area is a highly pertinent challenge within distributed robotics, as highlighted in \cite{HABFM02,HABFM03}. More recently, this issue has been framed in the context of graphs in the following manner: In a scenario where $k$ robots are initially situated on the nodes of an $n$-node graph, the robots undertake autonomous repositioning to achieve a final configuration wherein each robot occupies a distinct graph node (referred to as the {\em dispersion} problem) 
\cite{AM18}. 
This problem holds practical significance across various applications, such as the repositioning of self-driving electric cars (analogous to robots) to available charging stations (equivalent to nodes). This assumption involves the cars utilizing intelligent communication methods to locate unoccupied charging stations \cite{AM18,KA19}. Furthermore, the problem's importance stems from its interconnectedness with numerous other extensively researched challenges in autonomous robot coordination, including exploration, scattering, and load balancing \cite{AM18,KA19}.

The dispersion of mobile robots has garnered attention across various graph classes, including trees \cite{AM18,KA19}, rings \cite{AM18, KA19, MMM21}, arbitrary graphs \cite{AM18, KMS19, KMS20, KMS22, KS21, SSKM20}, dynamic graphs \cite{KMS20-dynamic}, directed graphs \cite{IPS22}. In the grid graph, the problem was explored by Kshemkalyani et al. \cite{KMS20}, but they considered oriented grid (called planar grid) and non-faulty robots exclusively. Orientation plays a pivotal role in the symmetric graphs, Barrière et. al. studied the scattering of autonomous mobile robots in the grid \cite{barriere2011uniform} and Becha et al. constructed a sense of direction by mobile robots in a torus \cite{becha2007optimal}.

{\bf Oriented vs Unoriented Grid:} In an oriented grid (see Figure~\ref{fig:planar_grid}), the ports are organized in such a way that allows a single robot to traverse the grid along a path with a clear sense of direction. When a robot enters a node via an incoming port, it simply needs to select the second port (i.e., leave one port after the incoming port and select the next port) as its outgoing port to continue moving in the same direction. In Figure~\ref{fig:planar_grid}, as a robot enters from port~1 at node $u$, it has the sense that the straight path leads from port~3 to node $v$. However, this straightforward approach is not applicable in an unoriented grid where the ports are interconnected arbitrarily, as depicted in Figure~\ref{fig:non_planar_grid}. In such a scenario, robots cannot distinguish whether they are moving in the same direction (i.e., along a row or column) or traversing in a cycle or zigzag manner across the unoriented grid. In Figure~\ref {fig:non_planar_grid}, a robot enters from the port~1 at node $x$. It is tough to decide based on the edges which path leads in the straight direction to $y$. In the grid, it appears that port~4 leads to the $y$ unlike port~3.
Consequently, it's not feasible to adapt the algorithm proposed in \cite{KMS20} to work in an unoriented grid.






\begin{figure}[htbp]
    \centering
    \begin{minipage}{0.45\textwidth}
        \centering
\begin{tikzpicture}[scale=1.3]
    \node[draw,circle,fill=blue,inner sep=0.07cm] (A1) at (0,0) {};
    \node[draw,circle,fill=blue,inner sep=0.07cm] (A2) at (1,0) {};
    \node[draw,circle,fill=blue,inner sep=0.07cm] (A3) at (2,0) {};
    \node[draw,circle,fill=blue,inner sep=0.07cm] (A4) at (3,0) {};
    
    \node[draw,circle,fill=blue,inner sep=0.07cm] (B1) at (0,1) {};
    \node[draw,circle,fill=blue,inner sep=0.07cm,label=below:$u$] (B2) at (1,1) {};
    \node[draw,circle,fill=blue,inner sep=0.07cm,label=below:$v$] (B3) at (2,1) {};
    \node[draw,circle,fill=blue,inner sep=0.07cm] (B4) at (3,1) {};
    
    \node[draw,circle,fill=blue,inner sep=0.07cm] (C1) at (0,2) {};
    \node[draw,circle,fill=blue,inner sep=0.07cm] (C2) at (1,2) {};
    \node[draw,circle,fill=blue,inner sep=0.07cm] (C3) at (2,2) {};
    \node[draw,circle,fill=blue,inner sep=0.07cm] (C4) at (3,2) {};
    
    \node[draw,circle,fill=blue,inner sep=0.07cm] (D1) at (0,3) {};
    \node[draw,circle,fill=blue,inner sep=0.07cm] (D2) at (1,3) {};
    \node[draw,circle,fill=blue,inner sep=0.07cm] (D3) at (2,3) {};
    \node[draw,circle,fill=blue,inner sep=0.07cm] (D4) at (3,3) {};
    
    \draw (A1) -- (A2);
    \draw (A2) -- (A3);
    \draw (A3) -- (A4);
    
    \draw (B1) -- (B2);
    \draw (B2) -- node[right = 10, above,red]{1}(B3);
    \draw (B3) -- (B4);

    \draw (B1) --node[right = 10, above,red] {1} (B2);
    \draw (B2) node[right = 10, above,red]{3}--(B3);
    \draw (B3) node[right = 10, above,red]{3}-- (B4);
    
    \draw (C1) -- (C2);
    \draw (C2) -- (C3);
    \draw (C3) -- (C4);
    
    \draw (D1) -- (D2);
    \draw (D2) -- (D3);
    \draw (D3) -- (D4);
    
    \draw (A1) -- (B1);
    \draw (A2) -- (B2);
    \draw (A3)node[right=5, above=16,red]{2} -- (B3);
    \draw (A4) -- (B4);

    \draw (A2) node[right=5, above=16,red]{2}-- (B2);
    \draw (A3) -- (B3);
    
    \draw (B1) -- (C1);
    \draw (B2)node[right, above=2, red]{4} -- (C2);
    \draw (B3) node[right, above=2, red]{4}-- (C3);
    \draw (B4) -- (C4);

    \draw (B2) -- (C2);
    \draw (B3) -- (C3);
    
    \draw (C1) -- (D1);
    \draw (C2) -- (D2);
    \draw (C3) -- (D3);
    \draw (C4) -- (D4);
\end{tikzpicture}

    \caption{$16$ nodes oriented square grid.}
    \label{fig:planar_grid}
    \end{minipage}
    \hfill
    \begin{minipage}{0.45\textwidth}
        \centering
       \begin{tikzpicture}[scale=1.1]
  \node[draw,circle,fill=blue,inner sep=0.07cm] (n00) at (0,0) {};
  \node[draw,circle,fill=blue,inner sep=0.07cm] (n01) at (0,1) {};
  \node[draw,circle,fill=blue,inner sep=0.07cm] (n02) at (0,2) {};
  \node[draw,circle,fill=blue,inner sep=0.07cm] (n03) at (0,3) {};

  \node[draw,circle,fill=blue,inner sep=0.07cm] (n10) at (1,0) {};
  \node[draw,circle,fill=blue,inner sep=0.07cm,label=below:$x$] (n11) at (1,1) {};
  \node[draw,circle,fill=blue,inner sep=0.07cm] (n12) at (1,2) {};
  \node[draw,circle,fill=blue,inner sep=0.07cm] (n13) at (1,3) {};

  \node[draw,circle,fill=blue,inner sep=0.07cm] (n20) at (2,0) {};
  \node[draw,circle,fill=blue,inner sep=0.07cm,label=below:$y$] (n21) at (2,1) {};
  \node[draw,circle,fill=blue,inner sep=0.07cm] (n22) at (2,2) {};
  \node[draw,circle,fill=blue,inner sep=0.07cm] (n23) at (2,3) {};

  \node[draw,circle,fill=blue,inner sep=0.07cm] (n30) at (3,0) {};
  \node[draw,circle,fill=blue,inner sep=0.07cm] (n31) at (3,1) {};
  \node[draw,circle,fill=blue,inner sep=0.07cm] (n32) at (3,2) {};
  \node[draw,circle,fill=blue,inner sep=0.07cm] (n33) at (3,3) {};

  \draw (n00) to[out=95, in=75] (n01);
  \draw (n00) to[out=225, in=95] (n10);
  \draw (n01) to[out=345, in=155] (n02);
  \draw (n01) to[out=265, in=275] node[right = 10, above,red]{1} (n11);
  \draw (n02) to[out=185, in=195] (n03);
  \draw (n02) to[out=140, in=115] (n12);
  \draw (n13) to[out=345, in=235] (n03);
  \draw (n10) to[out=245, in=355]node[above=10,red]{2} (n11);
  \draw (n10) to[out=145, in=275] (n20);
  \draw (n11) node[right, above =1 ,red]{4} to[out=415, in=295]  (n12);
  \draw (n11) node[right = 5,red]{3} to[out=125, in=215]  (n21);
  \draw (n11) node[right = 20,red]{1} to[out=125, in=215]  (n21);
  \draw (n12) to[out=435, in=235] (n13);
  \draw (n12) to[out=445, in=255] (n22);
  \draw (n13) to[out=455, in=275] (n23);
  \draw (n20) to[out=465, in=295] node[right=5,above=5,red]{2}(n21);
  \draw (n20) to[out=475, in=315] (n30);
  \draw (n21) node[left =2, above =3 ,red]{4}to[out=485, in=325] (n22);
  \draw (n21) node[right = 5, above,red]{3}to[out=145, in=345] (n31);
  \draw (n22) to[out=245, in=365] (n23);
  \draw (n22) to[out=345, in=385] (n32);
  \draw (n23) to[out=445, in=405] (n33);
  \draw (n30) to[out=45, in=425] (n31);
  \draw (n31) to[out=145, in=445] (n32);
  \draw (n32) to[out=245, in=4655] (n33);
\end{tikzpicture}

    \caption{$16$ nodes unoriented square grid.}
    \label{fig:non_planar_grid}
    \end{minipage}
\end{figure}

We provide a novel deterministic algorithm for dispersion on oriented grid graphs. Our algorithm works with faulty robots in an oriented grid graph, in which, a faulty robot can crash at any time during the execution of the protocol and never respond after crashing. Among the best-known results on dispersion \cite{KS21, CKMS23}, the paper \cite{KS21} considered non-faulty robots, and \cite{CKMS23, PSM23} considered faulty robots. However, applying the arbitrary graph results of \cite{KS21,CKMS23, PSM23} to grid graphs yields memory-optimal solutions but not time-optimal ones.
The dispersion problem in grid graphs was initially addressed by Kshemkalyani et al. in \cite{KMS20}, focusing on an oriented grid (termed as a planar grid) with non-faulty robots. This unidirectional movement allows for a solution to the dispersion problem in $O(\min(k, \sqrt{n}))$ time, utilizing $O(\log k)$ bits of memory per robot in the local communication model. Recently, \cite{BKM24} examined the dispersion problem on an unoriented grid, accounting for both faulty and non-faulty robots within the local communication model. This paper investigates the dispersion problem in oriented grids with faulty robots. Our results and a comparison with the closely related work \cite{KMS20, BKM24} are given in Table~\ref{tab:results} for a quick reference.

\begin{table*}[t]
\centering
\begin{tabular}{|m{3cm}|m{1.8cm}|m{1.8cm}|m{1.6cm}|m{1.8cm}|}
\hline
\multicolumn{5}{|c|}{\sc Comparative Analysis of Dispersion on Grid Graph} \\
\hline
Algorithm & Grid Type & Fault Type & Time (in rounds) & Memory (in bits)\\
\hline
Kshemkalyani et al. \cite{KMS20} & Oriented & No Fault & $O(\sqrt{n})$ & $O(\log k)$\\
Banerjee et al.~\cite{BKM24} & Unoriented & No Fault & $O(\sqrt{n})$ & $O(\log n)$\\
Banerjee et al.~\cite{BKM24} & Unoriented & Crash & $O(\sqrt{n}\log n)$ & $O(\sqrt{n} \log n)$\\
\textbf{Section~\ref{sec: faulty-planar}} & Oriented & Crash & $O(\sqrt{n})$ & $O(\log n)$\\
\textbf{Section~\ref{sec: faulty-planar}} & Oriented & Byzantine & $O(\sqrt{n})$ & $O(n\log n)$\\
\hline
\end{tabular}
\caption{Dispersion results for $n$ robots on an $n$-node square grid. In the faulty setup, up to $f \leq n$ robots may be faulty.}
\label{tab:results}
\end{table*}

\noindent \textbf{Our Contributions.} We present algorithms for the dispersion of $n$ robots in any arbitrary anonymous square grid in faulty robots' dispersion in the oriented grid.  The algorithm is deterministic. Specifically, we have the following contributions.
\begin{itemize}
    \item \textbf{Crash-Fault robots' dispersion on an oriented grid:} We present dispersion algorithm of $n$ robots in an oriented square grid of $n$ nodes having $f$ faulty robots such that $f \leq n$ that terminates in $O(\sqrt{n})$ rounds and uses memory bits $O(\log n)$ at each robot.

    \item \textbf{Byzantine-Fault robots' dispersion on an oriented grid:} We present  dispersion algorithm of $n$ robots in an oriented square grid of $n$ nodes having  $f$ Byzantine robots such that $f \leq n$ that terminates in $O(\sqrt{n})$ rounds and uses $O(n\log n)$ bits memory at each robot.
\end{itemize}

 Notice that $\Omega(\sqrt{n})$ is the trivial lower bound for round complexity for dispersing $n$ robots on an $n$-node grid. A robot requires the diameter of the graph time (round) to reach from one end to the other, where $2\sqrt{n}$ is the diameter of the grid. While $\Omega(\log n)$ is the lower bound for each robot's memory (in bits)\cite{AM18}. Therefore, our result is time and memory-optimal.

\medskip
\noindent \textbf{Challenges and Techniques.} Recall that robots in the oriented grid possess a sense of direction. Firstly, a robot can reach the boundary nodes (nodes with degree $3$) from the internal of the grid by moving in a straight direction. If the robots are faulty then gathering at a single corner becomes challenging due to the faulty nature of the robots. The single robot can crash and corner robots might wait for an indefinite time.

 We use two different techniques w.r.t to the value of $n$ to overcome these challenges, as follows. If $n$ is odd, then $\sqrt{n}$ is also odd. In that case, each robot reaches the middle of the boundary and then the center of the grid. Thus, all the robots gather at a single point of the grid. From there onwards, robots partition the grid into 4 equal parts and get dispersed. On the other hand, for the even value of $n$, there does not exist a single node as the center of the grid. In that case, each corner allows $n/4$ robots to stay and other robots are sent to the next corner. This approach ensures that there does not exist more than $n/4$ robot at a single corner. After that robot disperses from each corner in $n/4$-th of the grid having the dimension $\sqrt{n}/2 \times \sqrt{n}/2$ from each corner.

\noindent The rest of the paper is organized as follows.

\medskip

\noindent \textbf{Paper Organization.} Section~\ref{sec: model} states our distributed computing model. Section~\ref{sec: related_work} discusses the closely related work. Section~\ref{sec: faulty-planar} discusses the dispersion of faulty robots on the oriented grid. Finally, Section~\ref{sec: conclusion} concludes the paper with some interesting problems.
 
\section{Model of Computation}\label{sec: model}

\noindent{\bf Graph:} We consider an unweighted, undirected graph $G = (V, E)$ which is a square grid of $n=\sqrt{n}\times \sqrt{n}$ nodes embedded in $2$-dimensional Euclidean plane such that $\mid V \mid = n$ and $\mid E \mid = m$, where $V$ is the set of nodes and $E$ is the set of edges. $G$ is a connected graph having nodes with either degrees $2$ or $3$ or $4$. Nodes with degrees $2$, $3$, and $4$ are considered to be corner nodes, boundary nodes, and internal nodes, respectively. A square grid consists of $4$ corner nodes, $4\sqrt{n}-8$ boundary nodes, and $n-4\sqrt{n}+4$ internal nodes. These nodes are memoryless and resourceless means, unable to store any information and perform computation on them. Furthermore, nodes are anonymous such that nodes do not have IDs (identifiers) but each incident edge is uniquely identified by a labeled port number from $[1, \delta]$ where $\delta$ is the degree of the node. The nodes connected via an edge are termed neighboring nodes and are considered to be one hop away from each other. Additionally, each vertex of an edge is assigned a port number, and the port numbers of the two vertices on the same edge are independent of each other.\\

\noindent{\bf Robots:}  The set of robots $R = \{ r_1, r_2, \dots, r_n\}$ represents a collection of $n$ robots that are located across the graph $G$ at one or more nodes. Robots do not stay at the edge and stay only on the nodes of the graph $G$. Two or more robots situated at a node are called co-located robots and can communicate with each other. This model is known as {\em local communication model} \cite{AM18, KA19}. On the other hand, if the robots are allowed to communicate with any other robots in the graph (need not be co-located), the model is known as the {\em global communication model} \cite{KMS22,KMS20}. However, in this paper, we consider only the local communication model. Each robot contains a unique ID consisting of $O(\log n)$ bits. A robot can move from one node to another from the port if the nodes are connected to each other via an edge. Furthermore, it remains unaware of the other node of the graph until it visits that node. Each robot consists of some memory to store information and computation. We considered the computation time to be negligible as compared to the movement time of the robots from one node to another. A robot performs the "Communicate-Compute-Move" operation which is defined below.\\

\noindent{\bf Cycle:}   We consider a synchronous setting, every robot is synchronized to a common clock and movement from one node to another is complete in a one-time cycle or round. A robot $r_i \in R$ remains active in the "Communicate-Compute-Move" (CCM) cycle in a synchronous setting. Following are the three operations carried out by the robots:
\begin{itemize}
    \item {\bf Communicate:} $r_i$ can interact with the co-located robots and view the memory of a different robot, say $r_j \in R$.
    \item {\bf Compute:} $r_i$ can perform any required computation by using the data gathered during the "communicate" phase. This involves choosing a (potentially) port to leave $v_i$ and selecting the data to be saved in the robot $r_j$.
    \item {\bf Move:} $r_i$ writes new information (if any) in the memory of a robot $r_j$ at $v_i$,  and exits $v_i$ using the computed port to reach a neighbor node of $v_i$. Retaining a piece of information for a robot is considered writing it into its memory.
\end{itemize}

\noindent{\bf Crash Faults:} We consider the crash failure setup where a robot may fail by {\em crashing} at any time during the execution of the algorithm. The crashed robot is not recoverable and once a robot crashes it immediately loses all the information stored in itself, as if it was not present at all. Further, a crashed robot is not visible or sensible to other robots. We assume there are at most $f$ faulty robots such that $f\leq k$.\\

\noindent{\bf Time and Memory Complexity:} 
The time complexity conveys the number of discrete rounds or cycles taken before achieving dispersion. Memory complexity is the number of bits required to store each robot to successfully execute the algorithm. Our goal is to solve dispersion as fast as possible and keep the memory per robot as low as possible.\\

Below, we formally define 
faulty robots' dispersion problem in a graph. In this paper, we study the problem on an oriented grid graph.
\begin{definition}[\sc Faulty Robots Dispersion]
\label{def:FT-dispersion}
Given $n$ robots, up to $f\leq n$ of which are faulty, initially placed arbitrarily over the nodes of an $n$-node graph, the (non-faulty) robots re-position themselves autonomously such that each node has at most one non-faulty robot on it and subsequently terminate. 
\end{definition}

\begin{definition}[\sc Byzantine Robots Dispersion]
     Given $n$ robots, up to $f\leq n$ of which are Byzantine, initially placed arbitrarily on a graph of $n$ nodes, the non-Byzantine robots must re-position themselves autonomously to reach a configuration where each node has at most one non-Byzantine robot on it and subsequently terminate.
\end{definition}

\section{Related Work}\label{sec: related_work}
In this section, we discuss work related to deterministic algorithms for robot dispersion. In 2018, Moses Jr. et al. \cite{AM18} introduced the dispersion problem, where they presented a dispersion algorithm for different types of graphs. They argued a lower bound of $\Omega(\log n)$ for each robot's memory requirement. Later in 2019,  the lower bound was made more specific w.r.t. $k$ (number of robots) with $\Omega(\log (\max(k, \Delta)))$ in \cite{KMS19} where $\Delta$ is the highest degree of the graph. They showed that for an arbitrary graph, a deterministic algorithm requires $\Omega(D)$ rounds, where $D$ is the diameter of the graph. They also presented two techniques for the dispersion of robots on an arbitrary graph, one takes $O(\log n)$ memory and $O(mn)$ time while the other requires $O(n \log n)$ memory and $O(m)$ time.

In 2019, Kshemkalyani and Ali discussed several techniques for both synchronous and asynchronous models \cite{KA19}. They used $O(k \log \Delta)$ memory and $O(\min(m, k\Delta))$ rounds to solve the dispersion problem in the synchronous model. They presented many algorithms for the asynchronous model, one of which requires $O(\Delta D)$ rounds and $O(D \log \Delta)$ memory while another requires $O(\max(\log k,\log \Delta))$ memory with $O((m-n)k)$ time complexity. In the same year 2019, Kshemkalyani et al. improved the time complexity to $O(\min(m, k\Delta) \log k)$ at the cost of $O(\log n)$ memory \cite{KMS19}, but the knowledge of $m, n, k,$ and $\Delta$ were required by the robots in advance. In 2020, Takahiro et al. achieved the same complexity without the knowledge of the parameter $m, n, k,$ and $\Delta$ \cite{SSKM20}. Later in 2021, Kshemkalyani and Sharma improved the time complexity to $O(\min(m, k\Delta))$ with the graph-specific termination of the algorithm \cite{KS21}. 


In 2020, the dispersion problem is studied for faulty robot configurations. In \cite{MMM20}, Molla et al., dispersed the robots in  weak Byzantine settings (robots that behave randomly but cannot change their IDs) and examined the problem of unidentifiable rings. In \cite{MMM21} authors suggested several dispersion methods, some of which were tolerant of strong Byzantine robots (robots that behave arbitrarily and can even modify their IDs). Their algorithms are primarily based on the concept of assembling the robots at a root vertex, using them to create an isomorphic map of $G$, and then scattering them throughout $G$ by a predetermined protocol. 
Chand et. al. \cite{CKMS23} provided two algorithms on the mobile robots' dispersion on arbitrary graphs in the presence of crash fault robots with optimal memory, i.e., $O(\log (k+\Delta))$ bits. The first algorithm which has a rooted initial configuration has a time complexity of $O(k^2)$. On the other hand, the second algorithm which has an arbitrary initial configuration has a time complexity of $O((f + l) \cdot \min(m, k\Delta, k^2))$  when all the parameters are known to the robots, where $l$ is the cluster of robots in initial configuration. Recently, Pattanayank et al. \cite{PSM23} improved the result with the trade-off of one metric for the sake of other metrics, i.e., time and memory. They showed the time complexity $O(min\{m, k\Delta\})$ for $O(k \log(k+\Delta))$ bits complexity.
 
In the grid graph, the dispersion problem was first studied by Kshemkalyani et al. in \cite{KMS20}. They consider an oriented grid (which they called planar grid) and non-faulty robots. In an oriented grid a single robot can easily move through a path of the grid in the same direction using the ports' ordering at the nodes. This moving in one direction makes it simpler for the robots to solve the dispersion problem in $O(\min(k, \sqrt{n}))$ time with each robot requiring only $O(\log k)$ bits of memory in the local communication model. They also studied the problem in the global communication model, reducing the round complexity to $O(\sqrt{k})$ rounds. Another paper \cite{BKM24} studied the dispersion problem on an unoriented grid very recently, considering both faulty and non-faulty robots in the local communication model. In this paper, we study the dispersion problem in the oriented grid with faulty robots. A quick comparison of results is provided in Table~\ref{tab:results}.

\section{Dispersion of Faulty Robots on Oriented Grid}\label{sec: faulty-planar}
In this section, we present a deterministic algorithm for the dispersion of $n$ mobile robots, of which $f \leq n$ robots are faulty, on an oriented grid. Our primary goal is to minimize both the round and memory complexity, aiming to match the bounds of non-faulty robot dispersion on an oriented grid. In fact, our algorithm finishes in the optimal round and memory complexity: $O(\sqrt{n})$ rounds and $O(\log n)$ bits of memory per robot.

\subsection{Algorithm for Crash-Fault Robots}
The dispersion algorithm designed for non-faulty robots on an oriented grid in \cite{KMS20} cannot be adapted to work for faulty robots (crash-fault) in the same grid graph. Due to faulty nature of the robots, dispersion is more challenging as compared to non-faulty robots. The main idea of the algorithm in \cite{KMS20} is to first gather the robots at a single corner node of the grid, from which they can easily disperse. This gathering approach relies on using the robots' maximum and minimum IDs to break symmetry. However, this approach cannot be adopted in the faulty setup, as the robot with the maximum (or minimum) ID may crash during execution. To address this, we develop a new approach that allows the robots to operate independently in the faulty setup. Instead of gathering at a single corner node, depending on the size of the grid, the robots will: (i) either meet at the middle or center node in the case of an odd-sized grid, or (ii) partition the grid into four parts in the case of an even-sized grid and disperse. We refer to this algorithm as {\em Oriented Grid Dispersion with Faulty Robots}. Below, we explain the algorithm in detail. 
 
Note that the value of $n$ is not required to be known a priori by the robots; it can be computed during the execution of the algorithm. After reaching the corner nodes, the robots can move from one corner to another in one direction to determine the value of $\sqrt{n}$, and thus calculate $n$. This process will take an additional $\sqrt{n}$ rounds, which keeps the total round complexity unchanged. 
\medskip

\noindent \textbf{Gathering at the four corners of the grid:} A robot, located at the internal nodes, i.e., a node with degree $4$, initiates from the minimum port number from its location and reaches the boundary node, i.e, a node with degree $3$ of the grid by moving in a straight line or in the same direction. To move in one direction, a robot needs to select the second port (i.e., skip one port after the incoming port and choose the next port) as its outgoing port. It takes at most $\sqrt{n}-1$ rounds to reach the boundary of the grid. A robot at a boundary node moves along the boundary of the grid to reach a corner node. 
To do this, the robot selects the smallest port number from its starting location (node). If that port leads to the degree three nodes then move along the boundary in one direction until it eventually reaches a corner (i.e., a node with degree 2). Otherwise, the robot comes back to the boundary to select the smallest port number among the other two and moves along the boundary in one direction to reach a corner.
This also takes at most $\sqrt{n}$ rounds. Within $2\sqrt{n}$ rounds, all the robots gather at the four corners of the grid. After reaching the corner nodes, a robot moves to another in one direction to determine the value of $n$. Each robot requires $O(\log n)$ bits of memory to keep track of the number of rounds that have passed and their own ID. Note that each robot executes this process independently. 

\noindent \textbf{Dispersion for even number of node:} For the even-sized grid, the goal is to distribute the robots among the corners such that at most $n/4$ robots are present at a corner. After gathering at the corner nodes (previous stage), if a corner has more than $n/4$ robots, the algorithm keeps $n/4$ lowest ID robots at that corner and the rest of the robots move to another corner along the boundary. If any corner receives some extra robots from another corner then those robots move to the other corner. All other robots whose ID does not appear in first $n/4$ robots considered as extra robots. To move over the boundary of the grid and access all the corners takes $3\sqrt{n}$ rounds in the worst case. Therefore, after $6\sqrt{n}$ rounds ($2\sqrt{n}$ rounds used to reach the corner and $\sqrt{n}$ rounds to find the size of the grid), from the corner ($C_r$), $\frac{2C_r}{\sqrt{n}}$ (at most $\sqrt{n}/2$) robot move in each column with the help of the boundary nodes (including itself) in the direction of the minimum port number, if available. Each robot $r_u$ takes at most $6.5 \sqrt{n}$ rounds to reach their respective column based on their IDs, lower the ID nearer the column number. In the next $\sqrt{n}/2$ rounds, each robot disperses on the grid in their respective column in the order of lower to higher ID. Therefore, it takes $7\sqrt{n}$ rounds to disperse the robots in the square grid when $n$ is even. In this way, a corner covers at most $1/4^{th}$ grid, i.e., $\sqrt{n}/2 \times \sqrt{n}/2$ to disperse the $n/4$ robots (if available) at those nodes. 

\noindent \textbf{Dispersion for an odd number of node:} After $3\sqrt{n}$ rounds, corner robots move to the middle of the boundary nodes by initiating in the direction of the minimum port number, i.e., $\lfloor\sqrt{n}/2\rfloor$ hop away from the corner. It takes $\sqrt{n}/2$ rounds. After $3.5\sqrt{n}$ rounds ($2\sqrt{n}$ rounds used to reach the corner and $\sqrt{n}$ rounds to know the size of the graph), robots at the middle of the boundary move $\lfloor\sqrt{n}/2\rfloor$ hop away from the boundary node at $90\degree$ and reach the center of the grid in another $\sqrt{n}/2$ rounds. Since $n$ is odd, therefore, $\sqrt{n}$ is also odd. Notice that $\sqrt{n}$ is an integer. This implies, that there exists a unique node at equal distance between any two non-diagonal corner nodes. Furthermore, if any two robots are situated in the middle of the boundary nodes at opposite boundary then also there exists a unique middle node between them. Therefore, all the non-faulty robots meet at the center of the grid which is equidistant from all the corners. Henceforth, there exists a unique center node in the grid. After $4\sqrt{n}$ rounds, the centered robots collectively move to the corner of the grid as discussed above in \textit{gathering at the corners of the grid}. Notice that it takes $\sqrt{n}$ rounds to reach the corner of the grid from the center of the grid since robots move in straight lines and reach the middle of the boundary in $\sqrt{n}/2$ rounds and further reach at the corner of the grid in next $\sqrt{n}/2$ rounds. After $5\sqrt{n}$ rounds, 
from a single corner possesses $S_r$ many robots, $\frac{S_r}{\sqrt{n}}$ robot move in each column with the help of the boundary nodes (including itself) in the direction of the minimum port number, if available. Each robot $r_u$ takes at most $\sqrt{n}$ additional rounds to reach their respective corner based on their IDs, lower the ID nearer the column number. In the next $\sqrt{n}$ rounds, each robot is dispersed on the grid in their respective column in the order of lower to higher ID. Therefore, it takes overall $7\sqrt{n}$ rounds to disperse the robots in the square grid when $n$ is odd.

\begin{algorithm}[ht]
\caption{Oriented Grid Dispersion with Faulty Robots} \label{alg: faulty-oriented}
\footnotesize{
\SetKwInOut{Input}{Input}
\SetKwInOut{Output}{Output}

\Input{A square grid of $n = \sqrt{n}\times \sqrt{n}$ nodes, where $f$ robots are faulty such that $f \leq n$. The robots are distributed on the grid.}
\Output{Dispersion of the robots over the nodes.}

\BlankLine
Each robot $r_u$ (with degree $4$) reaches the boundary (node with degree $3$) of the grid by moving in a straight line.\tcp*{Takes $O(\sqrt{n})$ rounds}

    \If{$r_u$ is placed at a boundary node}{
        $r_u$ moves to the nearest corner (node with degree 2) of the grid \label{line: oriented_at_the_corner}\tcp*{Takes $O(\sqrt{n})$ rounds}
    }

    Each robot $r_u$ moves from one corner to another corner and learns the value of $\sqrt{n}$. Hence, $n$. \tcp*{Takes $O(\sqrt{n})$ rounds.}

\If{$n$ is even}{
    $n/4$ lower ID robots stay at its corner and rest remaining robots move to another corner\tcp*{Takes $O(\sqrt{n})$ rounds}\label{line: fining_corner_even}
    After $6\sqrt{n}$ rounds, from the corner with $C_r$ many robots, $\frac{2C_r}{\sqrt{n}}$ robots move in half ($\sqrt{n}/2$) of the column, if available \label{line: oriented_even_column}
}

\If{$n$ is odd}{
    After $3\sqrt{n}$ rounds, corner robots move to the middle of the boundary, i.e., $\lfloor\sqrt{n}/2\rfloor$ hops away from the corner\tcp*{Takes $O(\sqrt{n})$ rounds}
    After $3.5\sqrt{n}$ rounds, robots at the middle of the boundary move $\lfloor\sqrt{n}/2\rfloor$ hops away from the boundary at $90^\circ$ and reach the center of the grid
    \tcp*{Takes $O(\sqrt{n})$ rounds}
    After $4\sqrt{n}$ rounds, the centered robots move to the corner of the grid
    
    After $5\sqrt{n}$ rounds, from a single corner possesses $S_r$ many robots, $\frac{S_r}{\sqrt{n}}$ robots move in each column, if available \label{line: oriented_odd_column}
}

Each column disperses the robots across the column if available\tcp*{Takes $O(\sqrt{n})$ rounds}
All the robots settle at a unique node.
}
\end{algorithm}

From the above discussion, we have the following results.

\begin{theorem}\label{thm: crash_fault}
    Consider an oriented square grid with $n$ nodes, where $n$ robots are placed arbitrarily over the nodes, and $f$ of these robots are faulty, such that $f \leq  n$. Then, {\sc Faulty Robots Dispersion} can be solved deterministically in $O(\sqrt{n})$ rounds and $O(\log n)$ bits of memory per robot.
\end{theorem}
\begin{remark}
    It is easy to see that the result in Theorem~\ref{thm: crash_fault} also holds for $k \leq n$, where $k$ is the number of robots. 
\end{remark}

\subsection{Extension to Rectangular Grid}\label{sec: rectangular_grid}
We discussed the dispersion of faulty robots on the oriented square grid of $n = \sqrt{n} \times \sqrt{n}$ nodes in Algorithm~\ref{alg: faulty-oriented}. The Algorithm~\ref{alg: faulty-oriented} can be modified to the dispersion of the faulty robots on the rectangular grid of $n = \ell \times \frac{n}{\ell}$, where $\ell$ is the length and $n/\ell$ is the width of the rectangle such that $\ell, n/\ell>1$. In a square grid, length and width are of the same dimension, i.e., $\sqrt{n}$. In a rectangular grid, w.l.o.g, we consider that $\ell > n/\ell$. Henceforth, any operation that takes $\sqrt{n}$ rounds in a square grid would take at most $\ell$ rounds in a rectangular grid in every line of the Algorithm~\ref{alg: faulty-oriented}. After reaching at the corner (in line~\ref{line: oriented_at_the_corner}), it takes $\ell + n/\ell$ rounds to know the dimension of the grid. A robot $r_u$ can do this by moving from one corner to the diagonal corner. Therefore, the only change in the Algorithm~\ref{alg: faulty-oriented} is in the line~\ref{line: oriented_even_column} and line~\ref{line: oriented_odd_column} where $\frac{2C_r}{\ell}$ and $\frac{S_r}{\ell}$ robots are sent, respectively, across the longer column of the rectangular grid. This conveys that the rectangular grid's dispersion round and memory complexity would be $O(\ell)$ and $O(\log n)$, respectively. Since the diameter of the rectangular grid is $O(\ell)$, therefore, our algorithm is optimal w.r.t. round and memory complexity. The correctness and complexity proof directly follows from the square grid. Notice that this approach also does not require the dimension of the grid.



\subsection{Extension to Handle Weak Byzantine Robots}\label{app: weakly_Byzantine}
We discussed the dispersion of crash fault robots on an oriented square grid of $n = \sqrt{n} \times \sqrt{n}$ nodes in Algorithm~\ref{alg: faulty-oriented}. The Algortihm~\ref{alg: faulty-oriented} can be modified to work for a stronger adversary, i.e., a Byzantine adversary that controls the faulty (Byzantine) robots. Here we consider Weak Byzantine robots which can behave arbitrarily (e.g., deviate from the algorithm, pass wrong information, etc.) except to falsify their own identity (i.e., they cannot change or lie about their ID). 

In this Byzantine setting, the same Algorithm~\ref{alg: faulty-oriented} works correctly for an odd-sized grid. In the case of an even-sized grid, the algorithm~\ref{alg: faulty-oriented} can be extended to work but requires more memory for the robots to remember all the robots' IDs. For this, an $O(n \log n)$ bits of memory is required per robot. 
A robot $r_u$ with a higher ID settles at the next corner if the corner size (in terms of robots) is less than $n/4$ (excluding faulty robots in the knowledge of $r_u$), i.e., if a robot $r_u$ finds some robots have appeared at the previous corner then $r_u$ does not count those robots as the part of that particular corner. Recall that a robot has $n\log n$ memory, therefore, it keeps track of all the robots' that appeared in earlier locations. If the number of robots at a corner is less than $n/4$ then the Algorithm~\ref{alg: faulty-oriented} (in line~\ref{line: oriented_even_column}) executed as it is. For the case, when the number of robots is more than $n/4$ at a corner then the Lemma~\ref{lem: byzantine_even_case_settling} shows (and discusses the strategy due to space constraint) that a non-faulty robot $r_u$ does not settle with a non-faulty robot during the dispersion of Algorithm~\ref{alg: faulty-oriented} (in line~\ref{line: oriented_even_column}) in weak Byzantine setup.


\begin{lemma}\label{lem: byzantine_even_case_settling}
    In Byzantine dispersion, for an even value of $n$, a non-faulty robot $r_u$ does not settle with a non-faulty robot when the size of the corner is more than $n/4$.
\end{lemma}
\begin{proof}
    Let us assume $r_u$'s ID is not among the first $n/4$ IDs on the corner and the size of the corner is $n/4+x$ (including faulty and non-faulty robots). In that case, either $r_u$ was already at the corner or $r_u$ came from some other corner and decided to settle. In the first case, if $r_u$ was already at the corner, then all the other $x$ robots would declare with whom they would be settled (in increasing order of their ID) since they know at least $x$ faulty robots on the corner. In that case, $r_u$ considers two or more robots settling together as one and finds its settling location based on the index in $1/4$-th of the grid. If any incoming, robot does not declare the robot with which it would settle, then that robot is considered as faulty. Notice that a non-faulty robot would never settle with a non-faulty robot and a non-faulty robot would disclose the faulty robots in the knowledge. In the second case, if $r_u$ came from another corner then $r_u$ is well aware of which robot is suitable to settle. There exists at least one robot, say $r_v$, whose ID is in the first $n/4$ robots and non-accompanied by the non-faulty robot (in $r_u$'s knowledge). Let us suppose this is not the case and $r_v$ is the only robot accompanied by some non-faulty robot. In that case, there are only $x-1$ faulty robots known to $r_u$ and $r_u$ would not settle at the corner. 
    
    On the other hand, if the $r_u$'s ID is among the first $n/4$ IDs on the corner and even if the size of the corner is $n/4+x$ (including faulty and non-faulty robots). In that case, also, either $r_u$ was already the corner or $r_u$ came from some other corner and decided to settle. In both cases, $r_u$ settles on the grid as per Algorithm~\ref{alg: faulty-oriented}. Additionally, in the second case, $r_u$ provides the list of at least $x$ faulty robots. Then there exists at least one robot, say $r_v$, whose ID is not among the first $n/4$ IDs, therefore, $r_v$ settles with $r_u$. If $r_u$ is the honest then $r_v$ is the faulty for sure or at least one among them is faulty. If there is more than one robot that has a higher ID than the first $n/4$ robots listed by $r_u$ then based on ID and availability the robot settles with $r_u$. If $x\geq n/4$ then there exist at least $\lfloor 2x/n \rfloor$ robots together which are faulty due to the pigeonhole principle (known to some $r_u$). Therefore, one other robot can settle there.
    Hence, a non-faulty robot never settles with a non-faulty robot and hence the dispersion takes place.
\end{proof}

From the above discussion, we have the following results.

\begin{theorem}\label{Thm: Byzantine_dispersion}
   Consider an oriented square grid of $n$ nodes, in which $n$ robots are placed arbitrarily (over the nodes) and $f$ of them are weak Byzantine ($f\leq n$). If each robot has memory of $O(n \log n)$ bits, then there is a deterministic algorithm that solves dispersion in $O(\sqrt{n})$ rounds. 
\end{theorem}
\begin{remark}
    It is easy to see that the same algorithm (Algorithm~\ref{alg: faulty-oriented}) works for any $k$ robots' dispersion on an $n$-node grid, where $k\leq n$. Thus, the result in Theorem~\ref{Thm: Byzantine_dispersion} also holds for $k$ robots' dispersion, where $f\leq k$ are weak Byzantine. In addition, this weak-Byzantine dispersion result can be extended for the rectangular grid similar to the extension shown for crash fault robots in Section~\ref{sec: rectangular_grid}.
\end{remark}
 
\section{Conclusion and Future Work}\label{sec: conclusion}
In this paper, we studied fault-tolerant dispersion for distinguishable mobile robots on port-labeled square and rectangular grids. 
We designed and analyzed algorithms on oriented grids under two types of faults. For crash faults, our algorithm achieved a round complexity of $O(\sqrt{n})$ and a memory complexity of $O(\log n)$ bits per robot. For weak Byzantine faults, the same round complexity of $O(\sqrt{n})$ was maintained, while the memory requirement increased to $O(n \log n)$ bits per robot. Some open questions raised by our work: (i) Is it possible to achieve dispersion of $k \leq n$ robots in $O(\sqrt{k})$ rounds? (ii) Would similar bounds hold in the presence of strong Byzantine failures, where Byzantine robots could change their IDs.

\bibliographystyle{plain}
\bibliography{reference}


\end{document}